\documentclass[a4paper,12pt]{article}

\usepackage{amsthm}
\usepackage{amsmath}
\usepackage{epsfig}
\usepackage{amssymb}
\usepackage{latexsym}
\usepackage[T1]{fontenc}
\usepackage[utf8]{inputenc}
\usepackage{color}
\usepackage{url}

\newtheorem{theorem}{Theorem}

\newtheorem{lemma}{Lemma}
\newtheorem{conjecture}{Conjecture}

\let\epsilon=\varepsilon
\let\eps=\varepsilon
\let\rho = \varrho
\DeclareMathOperator{\tsp}{tsp}

\begin{document}

\title{Simple cubic graphs with no short\\ traveling salesman tour}

\author{
Robert Lukoťka \ and Ján Mazák
\\[3mm]
\\{\tt \{lukotka, mazak\}@dcs.fmph.uniba.sk}
\\[5mm]
Comenius University, Mlynská dolina, 842 48 Bratislava\\
}

\maketitle

\begin{abstract}
Let $\tsp(G)$ denote the length of a shortest travelling salesman tour in a graph $G$.
We prove that for any $\eps>0$, there exists
a simple $2$-connected planar cubic graph $G_1$ such that $\tsp(G_1)\ge (1.25-\eps)\cdot|V(G_1)|$, 
a simple $2$-connected bipartite cubic graph $G_2$ such that $\tsp(G_2)\ge (1.2-\eps)\cdot|V(G_2)|$,
and a simple $3$-connected cubic graph $G_3$ such that $\tsp(G_3)\ge (1.125-\eps)\cdot|V(G_3)|$.
\end{abstract}

\section{Traveling salesman problem in unweighted graphs}

The \emph{traveling salesman problem} (\emph{TSP}) is one of the most studied topics in both computer science and combinatorial optimization. 
Given a set of points and the distance between every pair of them, the task is to find the shortest Hamiltonian circuit.
\emph{Graphic TSP} is a special case of TSP: we are given an unweighted graph $G$ and want 
to find a shortest closed walk that contains all vertices of the graph.
Such walk is called a \emph{traveling salesman tour} or a \emph{TSP tour}. We will denote its length by $\tsp(G)$.

Recently, significant progress has been achieved in approximating graphic TSP. A series of incremental improvements culminated in a $1.4$-approximation algorithm of Seb\"o and Vygen \cite{SV}, providing the best ratio for general graphs known at present.
We refer the reader to the survey of Vygen \cite{V} for more information on approximation algorithms for TSP.

Together with general graphic TSP, several special instances of graphic TSP have been intensively studied. 
In particular, we are interested in \emph{cubic TSP} and \emph{subcubic TSP}, 
where the graph in question is cubic and subcubic, respectively.
(A \emph{cubic} graph is a graph with all vertices of degree three and a \emph{subcubic} graph is a graph with all vertices of degree at most three.) 
Note that every bridge (i.e., an edge whose removal increases the number of components of a graph) is used exactly twice in any TSP tour: it must be used at least twice and if it was used more, it would be possible to construct a shorter TSP tour.
After removing a bridge $b$ from a graph $G$ and finding TSP tours $T_1$ and $T_2$ of the resulting two components, 
we can create a TSP tour of $G$ by starting with $T_1$, then using $b$, continuing with $T_2$,
and closing the tour with $b$.
Thus to solve the problem of finding a TSP tour in cubic graphs with bridges, 
it suffices to solve the problem of finding a TSP tour in bridgeless subcubic graphs.
Since a graph must be connected to have a TSP tour, we are interested in $2$-connected subcubic graphs.

All present approximation algorithms for TSP on $2$-connected cubic and subcubic graphs only compare the length of the found tour with the number of vertices, that is, they basically do not care about the optimal TSP tour and just aim to find a TSP tour no longer than a constant multiple of the order of the given graph. We find the latter problem interesting in its own right.

M\"omke and Svensson \cite{MS} proved that every $2$-connected subcubic graph $G$ has a TSP tour of length at most 
$4/3 \cdot |V(G)|-2/3$. This bound is tight for both $2$-connected subcubic graphs and $2$-connected cubic multigraphs.
To improve the bounds further we have to examine a smaller class of graphs.
Simple $2$-connected cubic graphs seem to be the least restrictive class of graphs where improvement is possible.

Boyd et al. \cite{BSSS} proved that a simple $2$-connected cubic graph 
on $n$ vertices has a TSP tour of length at most $4/3 \cdot n -2$ (assuming $n \ge 6$). 
The result was subsequently improved by Correa, Larré, and Soto \cite{CLS} to $(4/3 - 1/61236) \cdot n$,
by van Zuylen \cite{Z} to $(4/3 - 1/8754) \cdot n$, by Candráková and Lukoťka \cite{CL} to
$1.3 \cdot n$, and very recently by Dvořák, Kráľ and Mohar \cite{DKM} to $9/7 \cdot n$.

For the class of subcubic graphs, graphs with no short TSP tours are easy to find. 
Indeed, it is sufficient to take two vertices and connect them by three equally long paths; one of the paths has to be used twice in any TSP tour.
This subcubic graph can be turned into a cubic graph by replacing vertices of degree $2$ by digons (cycles of length $2$); this replacement does not change the ratio of $4/3$ because whenever a TSP tour passes through a digon, it has to pass through all its vertices. This would fail, however, if we try to replace the vertices of degree $2$ with larger structures in order to avoid parallel edges.
Thus the situation is much less clear for the class of simple $2$-connected cubic graphs.
Boyd et al. \cite{BSSS} found an infinite family of $2$-connected cubic graphs
with no TSP tour shorter than $11/9 \cdot n - 8/9$.
We improve this lower bound by constructing, for any positive $\eps$, a $2$-connected cubic graph $G$ such that
$\tsp(G) > (1.25-\eps)|V(G)|$ (see Theorem~\ref{thmmain}).
We believe that the value $1.25$ is the best possible.

\begin{conjecture}\label{conj}
Every simple $2$-connected cubic graph with $n$ vertices has a traveling salesman tour of length at most $1.25n$.
\end{conjecture}

A similar construction and conjecture were recently independently proposed by Dvořák, Kráľ and Mohar \cite{DKM}.
For cyclically $4$-connected cubic graphs, Conjecture~\ref{conj} is implied by the dominating cycle conjecture of Fleischner \cite{fleischner}:
a dominating cycle $C$ in a cubic graph $G$ must contain at least $3/4$ of the vertices of $G$ and every other vertex has a neighbour on $C$, thus we can construct a walk with length at most $5/4\cdot |V(G)|$ visiting all the vertices.

TSP was also studied for cubic bipartite graphs.
The best present result (proved by van Zuylen in \cite{Z}) guarantees that a connected bipartite cubic graph different from $K_{3,3}$
has a TSP tour of length at most $1.25 n -2$.
We employ our method from Theorem~\ref{thmmain} to show that for any positive $\eps$, there exists a $2$-connected bipartite cubic graph $G$
such that $\tsp(G) > (1.2 - \eps)\cdot|V(G)|$ (see Theorem~\ref{thmmainbipartite}). 

Considering the rapid recent development in research on TSP for simple $2$-connected cubic graphs,
it is quite possible that Conjecture~\ref{conj} will be soon proved. 
After that, $3$-edge-connected graphs constitute a natural subclass of cubic graphs to study if we want to improve the bound on TSP length below $1.25 n$. In Theorem~\ref{thm3} we prove the existence of a family of $3$-edge-connected graphs that have TSP tours of lengths asymptotically approaching $1.125n$. 

%A graph is \emph{cyclically $k$-edge connected} if no cut of size at most $k-1$ separates two circuits. 
% For cubic graphs the notion cyclic edge-connectivity is a natural extension of the notion 
% edge-connectivity, which allows to differentiate cubic graphs that are $3$-edge-connected.
% No cyclically $8$-edge-connected cubic graph that is not Hamiltonian is known.
% Therefore it is perfectly possible that cyclically $8$-edge-connected cubic graphs all have 
% TSP tour of length $n$. Coxeter graph is one of few cyclically $7$-edge-connected
% non-Hamiltonian cubic graphs

%It would be especially interesting to look at cubic graphs that have no small cuts that separate circuits. There are various conjectures asserting that there are long or even Hamiltonian cycles in such graphs, but very little progress has been made so far (see \cite{longcycles} for an overview).

\section{Construction of graphs with no short tours}

Consider a TSP tour $T$ in a simple cubic graph $G$. 
No edge in $T$ is used three or more times, otherwise $T$ could be shortened. Thus the edges in $T$ are used either once or twice. 
The set of all the edges of $G$ that are used exactly once in $T$ induces a subgraph $H$ which is even (that is, all its vertices have even degree).
Actually, since $G$ is a cubic graph, all vertices of $H$ have degree $2$, and so $H$ is a collection of disjoint circuits.
The extension of $H$ by all the isolated vertices of $G-H$ yields an \emph{even factor} $F$ of $G$, that is, a subgraph which contains all vertices of $G$ and has even degrees of all vertices. Since $G$ is cubic, we can view its even factors as collections of disjoint circuits and isolated vertices.

Assume that the even factor $F$ corresponding to the tour $T$ contains $c$ circuits and $v$ isolated vertices.
We can easily express the length $|T|$ of $T$ in terms of $c$, $v$, and the order of $G$, denoted by $n$.
First, the edges of $T$ in $G/F$ (we contract all the circuits of $F$) constitute a spanning tree of $G/F$, otherwise $T$ could be shortened.
Each edge of this spanning tree is used twice in $T$, thus they contribute $2(c+v-1)$ to $|T|$.
The edges in circuits of $F$ contribute $n-v$. Altogether, $|T| = n-2 + 2c + v$. Our aim is to find graphs where the \emph{excess} $2c+v$ is large for every even factor.
For future reference, we denote by $q(G, F)$ the excess of an even factor $F$ in a graph $G$.

Our construction is based on suitable construction blocks.
Consider a gadget obtained from a $2$-edge-connected cubic graph $G$ by cutting an edge $e$ into a pair of dangling edges. Such structures are called cubic \emph{$2$-poles} (see \cite{NS, 2pole}).
The structure obtained from a $2$-edge-connected cubic graph $G$ by removing a vertex $v$ and replacing the three edges incident to $v$ by dangling edges is a \emph{$3$-pole}.
The notion of excess can be naturally extended to $r$-poles (we confine ourselves to $r\in \{2,3\}$ in the rest of this article). Even factors of an $r$-pole $P$ obtained from a graph $G$ are simply even factors of $G$ with the exception that whenever the original even factor $F$ contains an edge $e$ of $G$ that is not present in $P$, we replace $e$ by corresponding dangling edges. In particular, if $P$ is a $2$-pole, we replace $e$ by adding both dangling edges of $P$; if $P$ is a $3$-pole, we replace $e$ by the dangling edge sharing an endvertex with $e$. In both cases, two dangling edges are added and one circuit of $F$ is replaced by a path starting and ending with a dangling edge. This path contributes nothing to $q(P, F)$.

Every $r$-pole $P$ can be assigned a triple $t(P) = (q_0(P), q_2(P), n(P))$, where $q_0(P)$ is the minimum of $q(P, F_0)$ over all even factors $F_0$ of $P$ containing no dangling edges, $q_2(P)$ is the minimum of $q(P, F_2)$ over all even factors $F_2$ of $P$ containing two dangling edges, and $n(P)$ is the number of vertices in $P$. (The number of dangling edges belonging to an even factor is always even.) For instance, the $2$-pole arising from $K_4$ is assigned the triple $(2, 0, 4)$.

Given a $2$-pole $A$, we construct a $2$-pole $A'$ as indicated in Fig.~\ref{suciastka}: one dangling edge of each copy of $A$ is attached to a new vertex $x_1$, the other is attached to a new vertex $x_2$; in addition, $A'$ contains a path $x_1y_1y_2x_2$, one dangling edge incident to $y_1$, and one dangling edge incident to $y_2$. The relationship between $t(A')$ and $t(A)$ is captured in the following lemma.

\begin{figure}[b]
\centering\includegraphics{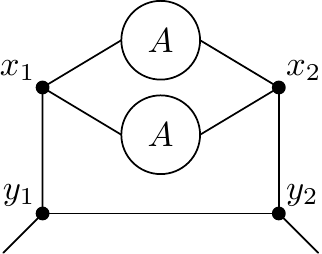}
\caption{The construction of $A'$ from $A$.}
\label{suciastka}
\end{figure}

\begin{lemma}
\label{lemmarecursive}
Let $A$ be a $2$-pole such that $t(A) = (a+2, a, n)$.
Then $$t(A') = (2a+4, 2a+2, 2n+4).$$
\end{lemma}

\begin{proof}
Let $F$ be an even factor in $A'$. First, we determine $q_2(A')$. If $F$ contains a path $C$ containing the dangling edges, there are two cases: either $C$ passes through $y_1$, $x_1$, one copy of $A$, $x_2$, and $y_2$, or it passes just through $y_1$ and $y_2$. In the first case, the copy of $A$ with nonempty intersection with $C$ contributes at least $q_2(A)$ to $q_2(A',F)$, and the other copy contributes at least $q_0(A)$, so $q_2(A')$ is $2a+2$.
In the second case, either $x_1$ and $x_2$ are isolated vertices in $F$, so contribute $2$, and each copy of $A$ contributes at least $q_0(A)$ to $q_2(A',F)$, or there is a circuit passing through $x_1$, $x_2$ and both copies of $A$, which contributes at least $2q_2(A)+2$ to $q_2(A',F)$. If we take the minimum, we get $q_2(A')=2a+2$.

Next, we consider $F$ that does not contain the dangling edges of $A'$ and determine $q_0(A')$. There are two cases: either $y_1$ and $y_2$ belong to a circuit $C$ (which must include $x_1$, $x_2$, and a copy of $A$), or $y_1$ and $y_2$ are isolated vertices in $F$. In the first case, the contribution to $q_0(A',F)$ is $2$ from the circuit $C$ and at least $q_0(A)+q_2(A)$ from the copies of $A$. In the second case, we have contribution of $2$ from $y_1$ and $y_2$ and at least $2a+2$ from the rest of $A'$ (the situation was analysed in the previous paragraph). Altogether, $q_0(A')=2a+4$.

Finally, $|V(A')|$ is obviously $2n+4$.
\end{proof}

We are ready to state and prove the main theorem.

\let\eps\varepsilon

\begin{theorem}
For any $\eps>0$, there exists a simple planar cubic graph $G$ with no traveling salesman tour of length shorter than $(1.25-\eps)|V(G)|$.
\label{thmmain}
\end{theorem}

\begin{proof}
We construct a sequence of $2$-poles $\{A_k\}_{k=0}^\infty$ with increasing ratio of excess 
to the number of vertices.

Let $A_0$ be the $2$-pole that arises from $K_4$ by cutting an edge. We have $t(A_0) = (2, 0, 4)$. 
For every integer $k\ge 1$, let $A_k = A_{k-1}'$. Note that if $A$ is planar, then $A'$ is planar, thus all the constructed $2$-poles are planar (since $A_0$ is planar).
According to Lemma~\ref{lemmarecursive},
the sequence of triples $\{t(A_k)\} = \{(a_k+2, a_k, n_k)\}$ satisfies 
\begin{eqnarray*}
a_k &=& 2a_{k-1} + 2,\quad a_0 = 0;\\
n_k &=& 2n_{k-1} + 4,\quad n_0 = 4.
\end{eqnarray*}
Solving the last two recurrence relations using the standard machinery of generating functions (see \cite{wilf}) yields $a_k=2\cdot 2^k-2$ and $n_k=8\cdot 2^k-4$ for every $k\ge 0$.

By inserting $A_k$ into an edge $e$ of a planar cubic graph $H$ isomorphic to $K_4$, we obtain a simple planar cubic graph $G_k$ with the ratio of the length of a traveling salesman tour 
to the number of vertices arbitrarily close to $5/4$.
Indeed, consider an even factor $F$ of $G_k$. If $F$ contains the edges arising from $e$, its excess is at least $q_2(A_k)+2$ (we add $2$ for the circuit passing through the vertices of $H$).
If $F$ does not contain those edges, it is composed of an even factor of $A_k$ with excess at least $q_0(A_k)$ and an even factor of $H$ with excess at least $2$.
Altogether, the length of a TSP tour $T$ is at least $n_k-2 + (a_k+2) = n_k+a_k$,
and thus the ratio $|T|/|V(G_k)|$ is equal to $(n_k+a_k)/(n_k+4)$, which belongs to $(1.25-\eps, 1.25)$ for a sufficiently large $k$.
\end{proof}

\begin{theorem}
For any $\eps>0$, there exists a simple bipartite cubic graph $G$ with no traveling salesman tour of length shorter than $(1.2-\eps)|V(G)|$.
\label{thmmainbipartite}
\end{theorem}

\begin{proof}
We construct a sequence of $2$-poles $\{A_k\}_{k=0}^\infty$ with increasing ratio of excess 
to the number of vertices. All the constructed 2-poles are {\it truly bipartite} in the following sense: they do not contain any odd circuit and every path both starting and ending with a dangling edge contains an even number of vertices.

Let $A_0$ be the truly bipartite $2$-pole that arises from $K_{3,3}$ by cutting an edge. We have $t(A_0) = (2, 0, 6)$. 
For every integer $k\ge 1$, let $A_k = A_{k-1}'$. Note that if $A$ is truly bipartite, then $A'$ is truly bipartite, thus all the constructed $2$-poles are truly bipartite.
According to Lemma~\ref{lemmarecursive},
the sequence of triples $\{t(A_k)\} = \{(a_k+2, a_k, n_k)\}$ satisfies 
\begin{eqnarray*}
a_k &=& 2a_{k-1} + 2,\quad a_0 = 0;\\
n_k &=& 2n_{k-1} + 4,\quad n_0 = 6,
\end{eqnarray*}
thus $a_k=2\cdot 2^k-2$ and $n_k=10\cdot 2^k-4$ for every $k\ge 0$.

By inserting $A_k$ into an edge of the bipartite cubic graph $K_{3,3}$, we obtain a simple bipartite cubic graph $G_k$ with the ratio of the length of a TSP tour  to the number of vertices being $(n_k+a_k)/(n_k+6)$, which approaches $1.2$ for large values of $k$.
\end{proof}

Next, we describe an analogous construction which yields $3$-connected simple cubic graphs with no short tours.
A $3$-pole $B$ is \emph{symmetric} if for each permutation $\pi$ of the three dangling edges of $B$ there exists 
an automorphism of $B$ that permutes the dangling edges according to $\pi$.
Given a symmetric $3$-pole $B$, we construct a symmetric $3$-pole $B''$ as indicated in Fig.~\ref{suciastka3}: 
we take the Petersen graph, remove a vertex, and replace each of the remaining vertices by a copy of $B$. 
Note that since $B$ is symmetric, the resulting $B''$ is uniquely determined and also symmetric. 
The relationship between $t(B'')$ and $t(B)$ is captured in the following lemma.

\begin{figure}[b]
\centering\includegraphics{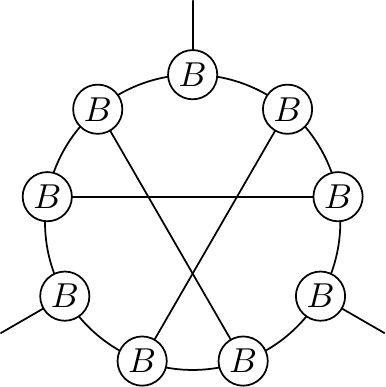}
\caption{The construction of $B''$ from $B$.}
\label{suciastka3}
\end{figure}

\begin{lemma}
\label{lemmarecursive3connected}
Let $B$ be a symmetric $3$-pole such that $t(B) = (b+1, b, n)$.
Then $$t(B'') = (9b+2, 9b+1, 9n).$$
\end{lemma}

\begin{proof}
First, we determine $q_0(B'')$. If there is a circuit $C$ of an even factor $F$ passing through at least $2$ copies of $B$, then $C$ contributes $2$ and 
each copy of $B$ contributes at least $b$ to the value of
$q_0(B'',F)$. If there is no such circuit, then each copy of $B$ contributes at least $b+1$ to $q_0(B'',F)$. Consequently, $q_0(B'') \ge 9b+2$.
Since $B$ is symmetric, for every choice of two dangling edges of $B$, there is a $2$-factor $F$ containing the chosen dangling edges with $q_2(B,F) = q_2(B)$.
It is, therefore, possible to get $q_0(B'', F) = 9b+2$ for a $2$-factor $F$ of $B''$ containing a circuit passing through all the nine copies of $B$ in $B''$, hence $q_0(B'') = 9b+2$.

Next, we determine $q_2(B'')$. Any path $C$ containing two dangling edges can visit at most $8$ copies of $B$ (otherwise the Petersen graph would be Hamiltonian; equality is attainable since $B$ is symmetric).
There is thus at least one copy not visited by $C$, which contributes at least $b+1$, and each of the remaining copies contributes at least $b$ to $q_2(B'',F)$. Altogether, $q_2(B'')=9b+1$.

Finally, $B''$ has obviously $9n$ vertices.
\end{proof}

\begin{theorem}
For any $\eps>0$, there exists a $3$-connected simple cubic graph $G$ with no traveling salesman tour of length shorter than $(1.125-\eps)|V(G)|$.
\label{thm3}
\end{theorem}

\begin{proof}
We construct a sequence of symmetric $3$-poles $\{B_k\}_{k=0}^\infty$ with increasing ratio of excess to number of vertices.

Let $B_0$ be a vertex incident to three dangling edges; we have $t(B_0) = (1, 0, 1)$. 
For every integer $k\ge 1$, let $B_k = B_{k-1}''$. According to Lemma~\ref{lemmarecursive3connected},
the sequence of triples $\{t(B_k)\} = \{(b_k+1, b_k, n_k)\}$ satisfies 
\begin{eqnarray*}
b_k &=& 9b_{k-1} + 1,\quad b_0 = 0;\\
n_k &=& 9n_{k-1},\quad n_0 = 1.
\end{eqnarray*}
Solving the last two recurrence relations by standard methods yields $b_k=(9^k-1)/8$ and $n_k=9^k$ for every $k\ge 0$.

By attaching a new vertex $v$ to the three dangling edges of $B_k$, we obtain a cubic graph $G_k$ with the ratio of the length of a TSP tour to the number of vertices arbitrarily close to $1.125$.
Indeed, consider an even factor $F$ of $G_k$. If $F$ contains two of the edges incident to $v$, its excess is at least $q_2(B_k)+2$ (we add $2$ for the circuit passing through $v$).
If $F$ contains $v$ as an isolated vertex, its excess is $q_0(B_k)+1$. Put together, $q(G_k, F) = \min\{b_k+2, (b_k+1)+1\} = b_k+2$.
Consequently, the length of a TSP tour $T$ is at least $n_k-2 + (b_k+2) = n_k+b_k$,
and thus the ratio $|T|/|V(G_K)|$ is equal to $(n_k+b_k)/(n_k+1)$, which belongs to $(1.125-\eps, 1.125)$ for a sufficiently large $k$.

\end{proof}

\noindent\textbf{Acknowledgements.}
We are thankful to the anonymous referees for useful suggestions for improving this paper, especially for pointing out that our construction principle is also applicable to bipartite graphs. We acknowledge support from the research grants VEGA 1/0474/15, VEGA 1/0876/16, and APVV-15-0220.

\end{document}